\documentclass[11pt]{article}

\usepackage[T1]{fontenc}
\usepackage{geometry}
\usepackage[round]{natbib}
\usepackage{times}
\usepackage{soul}
\usepackage{url}
\usepackage[hidelinks]{hyperref}
\usepackage[utf8]{inputenc}
\usepackage[small]{caption}
\usepackage{graphicx}
\usepackage{amsmath}
\usepackage{amsthm}
\usepackage{amssymb}
\usepackage{thm-restate}
\usepackage{booktabs}
\usepackage[linesnumbered,ruled]{algorithm2e} 
\usepackage[switch]{lineno}
\usepackage{microtype}
\usepackage{xspace}
\usepackage{cleveref}
\usepackage{dsfont}
\usepackage[colorinlistoftodos,obeyFinal]{todonotes}
\usepackage{enumitem}

\title{On Middle Grounds for Preference Statements}

\author{
Anne-Marie George, Ana Ozaki \\
University of Oslo, Norway \\
\texttt{\{annemage, anaoz\}@uio.no}
}

\date{}

\newtheorem{theorem}{Theorem}
\newtheorem{example}{Example}
\newtheorem{corollary}[theorem]{Corollary}%
\newtheorem{remark}{Remark}%

\newtheorem{definition}{Definition}%

\usepackage{xspace}

\newcommand{\calL}{\mathcal{L}}

\newcommand{\minv}[1]{\ensuremath{\mathsf{min}}\xspace}
\newcommand{\maxv}[1]{\ensuremath{\mathsf{max}}\xspace}
\newcommand{\Lmc}{\ensuremath{\mathcal{L}}\xspace}
\newcommand{\variable}{\ensuremath{v}\xspace}
\def\und{\underline}
\newcommand{\outc}{\und{V}}

\begin{document}

\maketitle

\begin{abstract}
    In group decisions or deliberations, stakeholders are often confronted with conflicting opinions. 
    We investigate a logic-based way of expressing such opinions and a formal general notion of a middle ground between stakeholders.
Inspired by the literature on preferences with hierarchical and lexicographic models, we instantiate our general framework to the case where stakeholders express their opinions using preference statements of the form \emph{I prefer `a' to `b'}, where \emph{`a'} and \emph{`b'} are alternatives expressed over some attributes, e.g., in a trolley problem, one can express \emph{I prefer to save 1 adult and 1 child to 2 adults (and 0 children)}. 
We prove theoretical results on the existence and uniqueness of middle grounds. In particular, we show that,  {for preference statements},
middle grounds may not exist and may not be unique. We also provide algorithms for deciding the existence and finding middle grounds. 
\end{abstract}

\section{Introduction}

High stake decisions or moral dilemmas, such as medical triage or the trolley problem, may prompt stakeholders to have strong opinions with little flexibility. 
The need to solve such decisions in real life requires the deliberation and consolidation of such possibly conflicting opinions. 
In this paper, we aim to break down stakeholder statements (e.g. statements about their moral preferences) into an agreeable set of statements --- a \textit{middle ground}.
Efforts in defining such a notion of a middle ground have recently been made by \citeauthor{DBLP:journals/aamas/OzakiRS24}~(\citeyear{DBLP:journals/aamas/OzakiRS24}).
However, their notion is designed
for Horn logic.
We propose 
a notion of middle ground for {a generic logic formalized as a satisfaction system~\citep{DBLP:journals/ai/AiguierABH18} 
and provide a case study for 
a logic that 
expresses}
\emph{preferences}.

Finding a middle ground between stakeholders can be an important first step to understanding and creating solutions
for conflicting opinions.
Applications of our work are thus manifold.
\citeauthor{freedman2020adapting}~(\citeyear{freedman2020adapting}), e.g., investigate human values in 
kidney exchanges, where patients are described by features of age, health, and drinking behaviour.
Unsurprisingly, the 289 participants of their survey did not agree on the prioritisation of patients.
Many other real-life scenarios may provoke 
{conflicting}
opinions or values: end-of-life medical decisions, decisions prompting trade-offs between economic advantages and preservation of nature, or hiring where the roles in a hiring committee 
warrant emphasis of different applicant features. 
The application 
{scenarios above often}
include stakeholders {who}
express preferences over alternatives.
We concentrate our case study on satisfaction systems similar to those described by~\citeauthor{WGB15}~(\citeyear{WGB15}), with
a language of comparative statements of the form ``\textit{I prefer a to b.}", where alternatives $a$ and $b$ are vectors of values from given variable domains. Models are then  lexicographic or hierarchical orders, i.e., total pre-orders on the set of alternatives. 
{That is, we  assume   stakeholders have (unknown) orders of importance for the features, by which they compare alternatives.} 
These satisfaction systems
transfer well to the Moral Machine Experiment~\citep{moralmachine}.
\begin{example}\label{ex:participant}
In the Moral Machine Experiment, participants (stakeholders) are asked to choose one out of two groups of individuals (alternatives) to save from a car accident.
The participant's choices can be interpreted as comparative preference statements 
like ``\textit{I prefer saving 1 adult, 4 children, and 0 dogs, to saving 2 adults, 3 children, and 3 dogs.}", in symbols, 
$(1, 4,0) > (2,3, 3)$.
This could, e.g., be modelled by the lexicographic model $(\mathsf{child}, \mathsf{adult},\mathsf{dog})$, which prioritizes children, over adults, and adults over dogs, or by the hierarchical model 
$(\{\mathsf{adult,child}\},\mathsf{child})$
 where alternatives are first compared on the number of humans, and only if they are equal  (there are 5 humans in both groups),
is the number of children  considered (4 in the first group, 3 in the second). 
The number of dogs 
is disregarded in the second model.
\end{example}

{Inspired by these scenarios,} this paper contributes with the following theoretical results.
\paragraph{General Notion of Middle Ground (MG):} We provide a general definition of \textit{middle ground} for satisfaction systems (Section~\ref{sec:mg-def}), show conditions for existence of a MG (Section~\ref{sec:mg-existence}) and an algorithm for construction (Section~\ref{sec:mg-construction}).
\paragraph{Case Study for Preference Statements:} We describe a satisfaction system similar to that of~\citeauthor{WGB15}(\citeyear{WGB15}) for modelling  preferences (Section~\ref{sec:hier-pref}), prove that  existence and uniqueness of a MG is not guaranteed under this system (Section~\ref{sec:hier-existence-and-uniqueness}), 
and complexity results of deciding the consistency of preferences and   existence of a MG (Section~\ref{sec:hier-algo}) for hierarchical models and the special case of lexicographic models.

\Cref{sec: conclusion} concludes. More proof details and discussions can be found in a longer version on Arxiv under the same title.

\section{Related Work}
{There has been several logic-based approaches   exploring  the   task of aggregating information and  resolving conflicts in different fields such as non-monotonic reasoning~\citep{Horty94,DelgrandeS97}, belief merging~\citep{Gardenfors1986}, argumentation~\citep{LiaoPST23}, ontology   repair~\citep{DBLP:conf/rr/MoodleyMV11}, and normative reasoning in ethical and legal contexts~\citep{Fengkui2020,Kolingbaum2008}. }

Our work is most similar to that by \citeauthor{DBLP:journals/aamas/OzakiRS24}~(\citeyear{DBLP:journals/aamas/OzakiRS24}) which defines a middle ground notion for Horn logic {and considers the Moral Machine Experiment}~\citep{moralmachine}.
However, while the postulates (P'1-P'6) in their definition explicitly use structural aspects of Horn expressions like the antecedent and consequent, we facilitate a more general definition with postulates (P1-5) for satisfaction systems. When interpreting their notion of coherence as the counterpart of consistency, then the two liken each other in spirit. The middle ground is a set of statement that is in it self coherent / consistent  (P'1 / P1), if possible equivalent to the union of all stakeholders' statements (P'2 / P2) and otherwise at least not in direct opposition to a stakeholders statement (P'3 / P3). Further, all statements in the middle ground should be motivated by  stakeholder statements and retain their statements as close as possible (P'4-6 / P4-5).
The work by
\citeauthor{KoniecznyP11}~(\citeyear{KoniecznyP11})
in belief merging contains some postulates that resemble the  notion of a middle ground. There,  an operator takes possibly conflicting beliefs from multiple sources as input and returns the belief base that is closest to the input and some integrity constraints. This differs from middle ground in that they  require  more properties to hold. In particular, integrity constraints expressed in propositional logic need to be satisfied and the operator needs  to compute \emph{the} closest belief base (which may not exist or be unique for middle ground).
Within social choice theory, \citeauthor{Botan2021}~(\citeyear{Botan2021}) investigate egalitarianism in judgement aggregation using propositional logic.
\citeauthor{Adler_2016}~(\citeyear{Adler_2016}) considers preference aggregation, arguing that preferences are more suitable than judgment for moral aggregation.
As a fundamental difference, a middle ground might be insufficient on its own for subsequent decision making but maintains some agreement of all stakeholders that a compromise or aggregation found by means of social choice methods cannot facilitate.

Previous attempts to model preferences include weighted sums over features (which are restrictive w.r.t. to the nature of such features)\citep{wilson2016preference}, Pareto models which lead to only partial orders~\citep{george2016preference}, and perhaps most convincingly but also less tractable Conditional Preference Networks~\citep{boutilier2004cp} and
{the  expressive prototypical preference logic}~\citep{bienvenu2010preference}.
Here, we lean our case study of preference statements onto the satisfaction systems described in~\citep{WGB15}. Preferences are modelled by some kind of hierarchical
models which are represented 
by
importance orders on variables/  features of alternatives. 
One drawback of these models is that they require variables to be non-repeating in the importance order. 
Instead, we consider 
models that are non-empty and allow for repeating variables at several importance levels.

\section{Middle Grounds for Sets of Statements}\label{sec:mg}
In this section we consider a general notion of middle ground for sets of statements and establish sufficient conditions for its existence. To make the presentation as general as possible, we first recall the notion of a
\emph{satisfaction system}~\citep{DBLP:journals/ai/AiguierABH18,DBLP:journals/jacm/DelgrandePW18,DBLP:conf/aaai/0001OR23}.

\newcommand{\mSet}{\ensuremath{\mathfrak{M}}\xspace}
\newcommand{\Bmc}{\ensuremath{\mathcal{B}}\xspace}
\begin{definition}[Satisfaction System] \label{def:satisfactionsystem}
A satisfaction system is a triple $(\Lmc,\models,\mSet)$, where 
\Lmc is a language, \mSet a set of models, and $\models$ a satisfaction relation on $\mSet\times\Lmc$.
The relation $\models$
  contains pairs of the form $(\pi, \phi)$
  with model $\pi$ \emph{satisfying} $\phi$. 
  
  Given $\Phi\subseteq\Lmc$, $\pi\models\Phi$ iff   $\pi\models \phi$ for all $\phi\in\Phi$.
Given $\Phi,\Phi'\subseteq\Lmc$,
we say that $\Phi$ \emph{entails} $\Phi'$, written $\Phi\models\Phi'$
if, for all $\pi\in\mSet$,
 $\pi\models \Phi$ implies $\pi\models \Phi'$. 
 Let $\mathsf{mod}(\Phi)$ denote the set of models that satisfy $\Phi\subseteq\Lmc$.

\end{definition}
Satisfaction systems have the following properties~\citep{DBLP:journals/ai/AiguierABH18}: if $\Phi\subseteq\Phi'$ then (1)
$\mathsf{mod}(\Phi')\subseteq\mathsf{mod}(\Phi)$; and (2) $\Phi'\models\Phi$ (monotonicity).
In this work, we consider satisfaction systems with \emph{finite} \Lmc.
We may treat $\phi$ in \Lmc and singleton set $\{\phi\}$ interchangeably.
{Elements of \Lmc are called \emph{statements}.
A set of statements $\Phi\subseteq\Lmc$ is \emph{consistent}
if 
$\mathsf{mod}(\Phi)\neq \emptyset$ and \emph{falsifiable} if $\mathsf{mod}(\Phi)\neq\mSet$. 
Also, 
$\Phi$  is \emph{non-trivial}
if it is consistent and falsifiable.
}

{Throughout this section, we consider an arbitrary satisfaction system $(\Lmc,\models,\mSet)$ and omit explicit references to it.}

\subsection{Notion of Middle Ground}\label{sec:mg-def}

Before we provide a formal definition of middle ground, we motivate it by considering a scenario where stakeholders have conflicting statements. 
    Recall \Cref{ex:participant}, 
    suppose another participant prefers the second alternative, that is, to save 2 adults, 3 children, and 3 dogs, in symbols,
$(1, 4,0)< (2,3, 3)$.
Is there a middle ground for these two participants? The union of their preferences is clearly inconsistent but
perhaps by ``weakening'' the second alternative, e.g., to $(2,3, 0)$
and also making the preference of the first participant non-strict,
{we can find an agreeable statement.}
That is, intuitively, $(1, 4,0)\geq (2,3, 0)$ is ``between'' the preferences of both participants.
This intuition is what we aim at capturing with  
middle grounds.

\begin{definition}[Middle Ground]\label{def:consensus}
{Let}
$\Phi_1,\ldots,\Phi_n$   be  
non-trivial sets of 
statements, 
each associated with a stakeholder 
$i\in\{1,\ldots,n\}$.

A set of statements $\Phi$ is a {\bf{middle ground}} for $\Phi_1,\ldots,\Phi_n$ 
if it 
satisfies each of the following postulates: 
\begin{itemize}[leftmargin=7.5mm]
\item[(P1)]  $\Phi$ is non-trivial;  
\item[(P2)] if  $\bigcup^n_{i=1} \Phi_{i}$ is consistent, then $\Phi  \equiv \bigcup^n_{i=1} \Phi_{i}$; 
\item[(P3)] for each $\phi\in\Phi$ and for all $i \in\{1,\ldots,n\}$ and all $\phi_i\in\Phi_i$, 
there is {$\pi \in \mSet$} such that
$\pi\models\phi$ and
$\pi\models\phi_i$;
\item[(P4)] for each $\phi\in\Phi$, there is $i\in\{1,\ldots,n\}$ with
$\Phi_i\models\phi$;
\item[(P5)] there is no 
$\Phi'$ such that $\Phi'\models\Phi$ and
$\Phi\not\models\Phi'$ and $\Phi'$ satisfies (P1)-(P4). 
\end{itemize}
\end{definition}

Considering the postulates in turn, we give an intuition behind the formalisation. 
The first postulate, P1, merely expresses that the statements in the middle ground should in itself make sense and be non-trivial.
The second postulate, P2, expresses that whenever the stakeholders statements are not contradictory, the middle ground should simply consist of a union of their statements or a logical equivalent ($\equiv$).
P3 expresses that {any statement in} the middle ground should be consistent with 
any individual statement of any of the stakeholders. 
{Though,} the middle ground might still oppose a collection of stakeholder preferences.
{P4 demands that any statement in the middle ground is justified by a stakeholder who's statements demand it.}
This is to prevent adding unnecessary statements to the middle ground.
Finally, P5 ensures that  
among the sets of statements that satisfy P1-P4, the middle ground 
{is maximal in the sense that it}
cannot be implied by another (non-equivalent) set.

To check that a middle ground is well defined, we need to consider the case of consistent stakeholders. It is easy to see that their joint statements satisfy the middle ground postulates.

\begin{restatable}{proposition}{ptwo}\label{prop:postulatetwo}
If  $\bigcup^n_{i=1} \Phi_{i}$ is consistent, then $\bigcup^n_{i=1} \Phi_{i}$ is a middle ground (\Cref{def:consensus}), that is, it satisfies P1-P5.
\end{restatable}

\subsection{Existence of Middle Ground}\label{sec:mg-existence}

The satisfaction of P1-P4 is sufficient for the existence.
For this, we note that the $\models$-relation is transitive, i.e., for $\Phi \models \Phi'$ and $\Phi' \models \Phi''$ we have $\Phi \models \Phi''$. Thus, $\models$ is acyclic {for non-equivalent statements}, i.e., there exists no chain of non-equivalent sets of statements $\Phi^1, \dots, \Phi^k$ such that $\Phi^i \models \Phi^{i+1}$ for $i=1, \dots, k-1$ and $\Phi^k \models \Phi^1$. In consequence, since we assume \Lmc is finite, 
there exists a \textit{dominating} set $\Phi$ such that there exists no other {non-equivalent} set $\Phi' \models \Phi$. Restricting $\models$ to sets of statements that satisfy P1-P4 preserves this. 

{Using
this observation and \Cref{prop:postulatetwo}} the {following holds.} 
\begin{restatable}{proposition}{ponepthreepfour}\label{pr:P134-sufficient-for-existence}
Let $\Phi_1,\ldots,\Phi_n$ be non-trivial sets of statements.
    If there exists a 
    set of statements $\Phi$ that satisfies P1, P3, and P4 then a middle ground exists for $\Phi_1,\ldots,\Phi_n$.
\end{restatable}

Further, by using that $\pi \models \Phi$ implies $\pi \models \phi$ for $\phi \in \Phi$ and transitivity of $\models$, we can show that to check existence of a middle ground  it is sufficient to consider \textit{single} statements rather than \textit{sets} of statements.

\begin{restatable}{proposition}{deductables}\label{pr:deductables-sat-P3-and-P4}
Let $\Phi_1,\ldots,\Phi_n$ be non-trivial sets of statements.
    If there exists a set of statements $\Phi$ that satisfies P3 and P4 for $\Phi_1,\ldots,\Phi_n$
    then any statement $\varphi$ such that $\phi \models \varphi$ for some $\phi \in \Phi$ satisfies P3 and P4.    
\end{restatable}

Similarly, one can also show that any $\varphi$ with $\Phi \models \varphi$ satisfies P3 if $\Phi$ satisfies P3. The same is not true for P4. 

However, if their union is consistent then we can show that it satisfies P1-P4 and thus, since they individually satisfy P5, they must be logically equivalent. Thus, middle grounds are either equivalent or inconsistent together.

\begin{restatable}    
{proposition}{equivorinconsistent}\label{pr:mgs-equiv-or-inconsistent}
    Let $\Phi$ and $\Phi'$ be two sets of statements that are middle grounds for stakeholder statements $\Phi_1,\ldots,\Phi_n$.
    Then either $\Phi \equiv \Phi'$ or $\Phi \cup \Phi'$ is inconsistent.
 \end{restatable}

\subsection{Construction of Middle Grounds}\label{sec:mg-construction}
 We can show that we can construct a middle ground with the help of the following algorithm, if there exists one. While not computationally efficient in general, this algorithm exploits the result of \Cref{pr:deductables-sat-P3-and-P4} by only considering satisfaction of P1, P3 and P4 for single statements rather than sets. This makes Algorithm~\ref{algo:construct-middle-ground} tractable for cases in which consistency and deduction problems are efficient.

\begin{algorithm}[t]
\SetKwInOut{Input}{Input}\SetKwInOut{Output}{Output}
\Input{Non-trivial  statement sets $\Phi_1,\ldots,\Phi_n \subseteq \Lmc$ }
\Output{{Set of all middle grounds (up to equivalence).}}
\BlankLine
\lIf{$\bigcup^n_{i=1} \Phi_{i}$ is consistent}{
\textbf{return} $\{\bigcup^n_{i=1} \Phi_{i}\}$ \label{lin:inconsistent}}
$\Psi_{1} := \{\varphi \in \mathcal{L} \mid \varphi \text{ non-trivial} \}$ \label{lin:pone}\;
$\Psi_{3} := \{\varphi \in \mathcal{L} \mid \forall \varphi' \in \bigcup^n_{i=1} \Phi_{i} \colon \{\varphi,\varphi'\} \text{ consistent} \}$ \label{lin:pthree}\;
$\Psi_{4} := \{\varphi \in \mathcal{L} \mid \exists i \in \{1, \dots, n\} \colon \Phi_{i} \models \varphi\}$ \label{lin:pfour}\;
{\textbf{return} the  
set of all cardinality-maximal consistent subsets of $\Psi_{1} \cap \Psi_{3} \cap \Psi_{4}$  (possibly empty) \label{lin:all}\;}
\caption{Middle Ground for  Statements}\label{algo:construct-middle-ground}
\end{algorithm}

\begin{theorem}\label{thm:algmiddleground}    
{Algorithm~\ref{algo:construct-middle-ground} returns the (possible empty) set of all middle grounds (up to logical equivalence) for non-trivial sets of stakeholder statements.}
 \end{theorem}
 \begin{proof}[Proof Sketch]
     If $\bigcup^n_{i=1} \Phi_{i}$ is consistent, then by Line~\ref{lin:inconsistent} the algorithm returns $\bigcup^n_{i=1} \Phi_{i}$ which, by P2 is {the only} middle ground {(up to logical equivalence)}.
     Then, assume  $\bigcup^n_{i=1} \Phi_{i}$ is inconsistent.
    In Lines~\ref{lin:pone}-\ref{lin:pfour},  Algorithm~\ref{algo:construct-middle-ground} constructs the sets $\Psi_i$ of  $\phi\in\Lmc$ that, individually, satisfies $P_i$, with $i\in\{1,3,4\}$.
    We show that {$\Phi$ is a  middle ground iff 
    $\Phi$ is equivalent to a 
 cardinality-maximal consistent subset of $\Psi:=\Psi_{1}\cap\Psi_{3}\cap\Psi_{4}$, returned by Algorithm~\ref{algo:construct-middle-ground} (Line~\ref{lin:all}) (note that $\Psi$ can be empty).}

 One can argue with the help of the postulates P1-4 and \Cref{pr:deductables-sat-P3-and-P4}, that if $\Phi$ is a middle ground then $\Phi$ is equivalent to a consistent subset of $\Psi$. 
 Next, one can show that any consistent subset $\Phi$ of $\Psi$ satisfies P1-P4. 
 Note that elements of a set of statements satisfy P3 and P4 individually, then the set also satisfies P3 and P4.
 We can now show that any cardinality-maximal subset $\Phi$ of $\Psi$ that is consistent satisfies P5. Assume for contradiction that another set of statements $\Phi'$  satisfys P1-P4 and $\Phi' \models \Phi$ and $\Phi \not\models \Phi'$. One can now show $\Phi \cup \Phi'$ is inconsistent which implies $\Phi \not \models \Phi'$ --- a contradiction.

    We have shown that any middle ground is equivalent to a consistent subset of $\Psi$, and any cardinality-maximal consistent subset of $\Psi$ is a middle ground.
    By \Cref{pr:mgs-equiv-or-inconsistent}, any two middle grounds are either equivalent or their union is inconsistent. Now, any non-cardinality-maximal consistent subset $\Phi$ of $\Psi$ is consistent with one that is cardinality-maximal  and thus is a middle ground. 
    Hence,  any middle ground is equivalent to a cardinality-maximal consistent subset of $\Psi$.
 \end{proof}

\section{Middle Grounds for Preferences}\label{sec:mgpreferences}

Here,
we instantiate the general framework of  \Cref{sec:mg} 
for the satisfaction system $\Lambda$ with the language of preferences (\Cref{def:language}) and hierarchical models (\Cref{def:hierarchical-model}). We also consider the special case of $\Lambda$ where we 
{consider}
the class of lexicographic models (\Cref{def:lex-model}).
We start by formally defining all necessary notions 
and then analyse the complexity of deciding the existence of a middle ground.

\subsection{Hierarchical Preferences}\label{sec:hier-pref}
\paragraph{Variables and Alternatives:}
Let $V$ be a set of $m$ \textit{variables} (or features) which describe alternatives.
For each variable $\variable \in V$,  
let $\und{\variable}$
denote its \textit{domain}, i.e., the set of possible values of $\variable$.  
Assume that $\und{\variable}$ is finite and contains more than one element.
An \emph{alternative} is an element of {$\und{V} = \prod_{\variable \in V} \und{\variable}$}
i.e., an assignment to all the variables.
For alternative $\alpha \in \und{V}$ and variable $\variable \in V$,
let $\alpha(v) \in \und{v}$ be the value $\alpha$ assigns to $v$.

\begin{example}[cont.]\label{ex:varaibles-and-alternatives}
    As before, we consider a setting similar to that in the Moral Machine Experiment~\citep{moralmachine}. 
    More concretely, let the alternatives be described by three variables with {values between $0$ and $5$} as domains, such that
    $\und{V} = \und{\sf adult} \times \und{\sf child} 
    \times \und{\sf dog} $. 
    Consider the   alternatives:
    \[\alpha = (1,4,0), \quad \beta = (2,3,3), \quad \gamma = (1,3,5).\]
    Then, $\alpha$ describes a set of  1 adult, 4 children, and 0 dogs. Similarly, $\beta$ and $\gamma$ specify sets of adults, children,  and dogs. 
\end{example}

A hierarchical model consists of a hierarchy 
over variables.
At each level of the hierarchy, we combine the variable assignments by a commutative and associative operator $\bigoplus$.
Here, we assume that value domains of variables are compatible, i.e.,
 there exists an operator $\bigoplus$ that can combine any subset of variables in a meaningful way, and there exists a natural order relation over the value domains as well as over values of combinations of variables. 
We can then compare alternatives by a lexicographic order. That is, we compare alternatives first based on the value combinations of the first-level variables; only if these are equal is the combination of the next most important variables considered, 
and so on.

\begin{definition}[Hierarchical  Model]
\label{def:hierarchical-model}
A \emph{hierarchical model}, or simply model, $\pi$ over variables $V$,  is defined
to be a 
non-empty sequence
of the form
$(Y_1, \ldots, Y_k)$.
Here  $Y_1, \ldots, Y_k \subseteq V$ are $k$ non-empty sets of variables in $V$.
\end{definition}

\begin{definition}[Lexicographic Model]\label{def:lex-model}
A \emph{lexicographic model} is a hierarchical model with singleton variable sets. With an abuse of notation, we write such sequences as
$(v_1, \ldots, v_k)$, where
 $v_1, \ldots, v_k \in V$.
\end{definition}

Our definitions are very similar to the models defined by~\citeauthor{WGB15}~(\citeyear{WGB15}), but differ in two points. First, we assume that neither hierarchical nor lexicographic models can be empty sequences. The corner case of empty models is a technical detail but, as becomes clearer in the following, does not contribute meaningful inference of preference statements. A more important difference is that, by our definition, hierarchical models may have non-disjoint sets of variables. 
That is, it would  be possible to express that the number of humans is most important and the number of children is second most important, since children would appear in two levels of the importance order.
Our definition is, in this latter point, a generalisation of the models defined in~\citep{WGB15}.

For any hierarchical models together with a commutative and associative operator $\bigoplus$ we define an order relation $\succeq_\pi$ over alternatives (omitting $\bigoplus$ for readability).

\begin{definition}[Order Relation $\succeq _\pi$]\label{def:order}
Let $V$ be variables and $\bigoplus$ a commutative and associative   
    operator  on the variable domains. Assume that there exists a total order relation $\geq$ on the variable domains and on $\bigoplus$-combinations of variable values.
    For a model $\pi = (Y_1, \ldots, Y_k)$ over variables $V$ the binary relation $\succeq _\pi$ on
$\und{V}$ is defined as follows. \\
For alternatives $\alpha,\beta \in \und{V}$, we have
$\alpha \succeq _\pi \beta$ if and only if 

\begin{enumerate}[label=(\roman*)]
    \item for all $i=1, \ldots, k$,  $\bigoplus_{y \in Y_i} \alpha(y) = \bigoplus_{y \in Y_i}\beta(y)$, or
    \item there exists $i \in \{1, \ldots, k\}$ s.t. 
    \begin{itemize}
        \item $\bigoplus_{y \in Y_i} \alpha(y) > \bigoplus_{y \in Y_i}\beta(y)$ and
        \item $\bigoplus_{y \in Y_j} \alpha(y) = \bigoplus_{y \in Y_j}\beta(y)$ for all $j < i$.
    \end{itemize} 
\end{enumerate}
\end{definition}

The order relation $\succeq _\pi$ is a total pre-order on $\und{V}$, i.e., reflexive, transitive and total.
The order relation is not necessarily complete as it, e.g., does not necessarily include all variables. Thus, two alternatives might appear to be equivalent under $\succeq _\pi$ whereas they are different elements in $\und{V}$.

The corresponding strict relation $\succ_\pi$ is given by
$\alpha \succ_\pi \beta$ if and only if  (ii) is satisfied, i.e., there exists $i \in \{1, \ldots, k\}$ such that 
$\bigoplus_{y \in Y_i} \alpha(y)> \bigoplus_{y \in Y_i} \beta(y)$
and for all $j < i$, 
$\bigoplus_{y \in Y_j} \alpha(y) = \bigoplus_{y \in Y_j} \beta(y)$.
The corresponding equivalence relation $\equiv_\pi$
is given by
$\alpha \equiv_\pi \beta$ if and only if (i) is satisfied, i.e.,
for all $i=1, \ldots, k$,  $\alpha(Y_i) = \beta(Y_i)$.

\begin{example}[cont.]\label{ex:order-relation}
    Consider the 
    alternatives 
in~\Cref{ex:varaibles-and-alternatives}.
    If, 
    it is desirable to save as many living beings as possible, then the natural order 
    is ``the more the better''.      
    As the domains are compatible (they are all the same) we could for example take the usual addition as the operator $\bigoplus$.  
    Consider $\pi = (\{{\sf adult}, {\sf child}\},\{{\sf child}\})$. This hierarchical model expresses that the number of humans to be saved from a car crash is the most important. Only if they are equal do we consider the number of children.
    Dogs are irrelevant in the comparison.
    Under the order relation induced by $\pi$, we have 
    that $\alpha$ is strictly preferred to  $\gamma$ (and $\beta$), $\alpha \succ_\pi \gamma$ and thus $\alpha \not\preceq_\pi \gamma$.
\end{example}

\begin{definition}[Preference Language $\Lmc$,~\citep{WGB15}]\label{def:language}
We define the language of non-strict and strict preference
statements that are simple comparisons of alternatives $\outc$ as:
\[\Lmc = \{\alpha \geq \beta \mid \alpha, \beta \in \outc\} \cup \{\alpha > \beta \mid \alpha, \beta \in \outc\}\]
\end{definition}

{We add parenthesis around preference statements when they appear in sequence, e.g. $(\alpha \geq \beta), (\gamma < \delta)$,  for visualization.}
As \citeauthor{WGB15}~(\citeyear{WGB15}), we define the meaning of these statements, and a satisfaction relation $\models$ between hierarchical models $\pi$ and  statements, in correspondence to $\succeq _\pi$.

\begin{definition}[Satisfaction Relation $\models$]\label{def:satisfaction}
    Let $\pi$ be a hierarchical model, and  $\alpha, \beta \in \und{V}$ alternatives.
    
    \begin{itemize}
        \item We say that $\pi$ satisfies the non-strict statement $\alpha \geq \beta$, denoted by $\pi \models \alpha \geq \beta$, if and only if $\alpha \succeq _\pi \beta$.\\ That is, under $\pi$, $\alpha$ is at least as preferred as $\beta$.
    \item We say that $\pi$ satisfies the strict statement $\alpha > \beta$, denoted by $\pi \models \alpha > \beta$, if and only if $\alpha \succ _\pi \beta$.\\ That is, under $\pi$, $\alpha$ is strictly preferred to $\beta$.
    \end{itemize}
\end{definition}

Through $\Lmc$, a stakeholder can express indifference between $\alpha$ and $\beta$ via the statements $\alpha \geq \beta$ and $\beta\geq \alpha$ together. Further, as \citeauthor{WGB15}~(\citeyear{WGB15}) already state, because $\succeq _\pi$ is a total pre-order over the alternatives, $\pi \not\models \alpha \geq \beta$ is equivalent to $\pi \models \beta > \alpha$. For this reason, we omit the definition of negated statements in $\Lmc$.
The notion of entailment is as in \Cref{def:satisfactionsystem}.

\begin{example}[cont.] 
    The 
    statement $\beta > \alpha$, i.e., $\beta$ is strictly preferred to $\alpha$, intuitively implies that any model $\pi$ of the statement contains at least one set of individuals (that are \textit{important} to the stakeholder) from which there are strictly more beings saved in $\beta$ than in $\alpha$,
    e.g., 
    \((\{{\sf adult},{\sf child}\}, \{{\sf dog}\}) \models \beta > \alpha. \)
    \( \text{ While }(\{{\sf adult},{\sf child}\}, \{{\sf dog}\})\models \alpha > \gamma,\)
    we cannot  deduce  $ \alpha > \gamma$ from $\beta > \alpha$ ($\{\beta > \alpha\} \not \models \gamma > \alpha$) because also 
    $(\{{\sf dog}\}) \models \beta > \alpha$ and $(\{{\sf dog}\}) \not \models \alpha > \gamma$.
\end{example}

\subsection{Non-Uniqueness and Non-Existence}\label{sec:hier-existence-and-uniqueness}
  As a first result, we note that there may be more than one middle ground for preference statements in $\calL$.

\begin{restatable}
    {theorem}{nonunique}
    There exist sets of stakeholder statements in $\calL$ that admit multiple non-equivalent middle grounds.
 \end{restatable}
\begin{proof}[Proof Sketch]
Consider the following alternatives defined over four binary variables $V = \{x,y,z,w\}$:

    \begin{center}
    \begin{tabular}{c c c c c}
                        & $x$ & $y$ & $z$ & $w$\\ \hline
           $\alpha =$  & 1 & 0 & 0 & 0\\
           $\beta =$  & 0 & 1 & 0 & 0 \\
            $\alpha' =$  & 0 & 0 & 1 & 0\\
           $\beta' =$  & 0 & 0 & 0 & 1 \\
           $\gamma =$  & 1 & 0 & 1 & 0 \\
           $\delta =$  & 0 & 1 & 0 & 1 
        \end{tabular}
    \end{center}

     For simplicity, we assume that the value of any {$\oplus$-combination} of variables is the same for all alternatives and omit such values in the table on the left.
      Consider two stakeholders expressing non-trivial statements:
        \[\Phi_1 = \{(\alpha > \beta),(\alpha' > \beta')\}, \quad \Phi_2 = \{ (\beta > \alpha), (\beta' > \alpha')\}.\]
            The stakeholder's statements are consistent individually, but inconsistent together. Thus, the union of $\Phi_1$ and $\Phi_2$ cannot be a middle ground. One can show that there are at least two non-equivalent middle grounds for  $\Phi_1$ and $\Phi_2$ with help of the two statements $\psi_1 = \gamma > \delta$ and $\psi_2 =  \delta > \gamma.$
            In particular:
        \begin{enumerate}
            \item $\psi_1$ and $\psi_2$ are individually non-trivial;
            \item $\psi_1$ and $\psi_2$ are  inconsistent together; 
            \item for all $i,j\in \{1,2\}$ and all $\phi_i\in\Phi_i$, there is $\pi$ such that
            $\pi\models \psi_j$ and $\pi\models \phi_i$;
            \item for $i\in \{1,2\}$, $\Phi_i\models\psi_i$.
        \end{enumerate}               To conclude, we claim that
        there are at least two non-equivalent middle grounds: one that contains $\psi_1$ and another one that contains $\psi_2$.
        Indeed,
         Points (1), (3), (4) and \Cref{{thm:algmiddleground}} imply that that there is a middle ground for 
        $\Phi_1$ and $\Phi_2$ that contains
        $\psi_1$
        (plus possibly other statements, so as to satisfy P5) and  a middle ground for 
        $\Phi_1$ and $\Phi_2$ that contains
        $\psi_2$. Point (2) implies that
        there is no middle ground that contains both $\psi_1$ and $\psi_2$ (otherwise P1 would be violated). So there are   two non-equivalent middle grounds for $\Phi_1$ and $\Phi_2$.    
\end{proof}

{Further, we show that a middle ground may not exist.}
    \begin{theorem}
    There exist sets of stakeholder statements in $\calL$ that admit no middle ground.
 \end{theorem} 
\begin{proof}
Consider alternatives defined over two binary variables $V = \{x,y\}$, and an operator $\oplus$
  that resembles the logical $\wedge$:
  \begin{center}
  \begin{tabular}{c c c c}
                        & $x$ & $y$ & $x\oplus y$ \\ \hline
           $\alpha =$  & 1 & 0 & 0\\
           $\beta =$  & 0 & 1 & 0\\
           $\gamma =$  & 1 & 1 & 1\\
           $\delta =$  & 0 & 0 & 0
        \end{tabular}
  \end{center}
    By the convention $1 > 0$, any 
    {hierarchical}
    model entails $\alpha \geq \delta$, $\beta \geq \delta$ and $\gamma > \delta$, as well as $\gamma \geq \alpha$ and $\gamma \geq \beta$. Further, no model satisfies $\delta > \gamma$.
    The set of non-trivial statements in this case is given by $\mathcal{N} = \{(\gamma > \alpha), (\gamma \leq \alpha), (\gamma > \beta), (\gamma \leq \beta), (\alpha > \beta), (\alpha \geq \beta), (\alpha < \beta), (\alpha \leq \beta), (\alpha > \delta), (\alpha \leq \delta), (\beta > \delta), (\beta \leq \delta)\}$.
    
      Consider two stakeholders 
      with
      preference statements: 
        \[\Phi_1 = \{\alpha \geq \gamma\}  \quad  \text{and}\quad \Phi_2 = \{\beta \geq \gamma\}.\]
    {These statements}
    are consistent individually, but inconsistent together.
    In particular, the only hierarchical model that satisfies $\Phi_1$ is $(\{x\})$. 
    Thus $\Phi_1$ entails non-trivial statements $\{ (\gamma > \beta), (\alpha \geq \gamma), (\alpha > \beta), (\alpha \geq \beta), (\alpha > \delta), (\delta \geq \beta)\} \subseteq \mathcal{N}$.
    None of these statements is consistent with $\Phi_2$.
    By symmetry, the only hierarchical model satisfying $\Phi_2$ is $(\{y\})$ 
    and none of the entailed statements from $\Phi_2$ are consistent with $\Phi_1$.

    By P4 any statement in the middle ground is entailed by some stakeholders statements. However, by P3 and because the stakeholders have only one statement each, the middle ground needs to be consistent with the stakeholders statements. As argued above, there is no such middle ground. 
    \end{proof}

\subsection{Deciding Existence of a Middle Ground}
Through \Cref{pr:P134-sufficient-for-existence} we have established that for the existence of a middle ground it is sufficient to check whether there exists a set of statements that satisfies P1, P3, and P4. Further, we found that, by \Cref{pr:deductables-sat-P3-and-P4}, it is sufficient to only check for the existence of single (non-trivial) statements that satisfy P3 and P4.
For preference statements of language $\calL$ we can further narrow down which statements shall be investigated to determine existence of a middle ground.

As a consequence of \Cref{pr:deductables-sat-P3-and-P4}, and because $(\alpha > \beta) \models (\alpha \geq \beta)$, we have the following relation between strict and non-strict statements satisfying P3 and P4.

\begin{corollary}
Let $\Phi_1,\ldots,\Phi_n \subseteq \calL$ be non-trivial sets of statements.
    If the strict statement $\alpha > \beta$ satisfies P3 and P4 then its non-strict version $\alpha \geq \beta$ satisfies P3 and P4. 
\end{corollary}
Here the non-strict version, while satisfying P3 and P4, might be trivial (i.e., violating P1) even if the strict statement is non-trivial. 
However, 
we can observe that this can only happen in a specific case.

\begin{restatable}{lemma}{nontrivial}\label{le:non-strict-version-of-strict-non-trivial}
    If $\alpha > \beta$ is non-trivial then either $\alpha \geq \beta$ is non-trivial or 
    \begin{itemize}
        \item $\bigoplus_{y \in Y} \alpha(y) \geq \bigoplus_{y \in Y} \beta(y)$ for all $Y \subseteq V$, 
and 
        \item there exists $Y \subseteq V$  with $\bigoplus_{y \in Y} \alpha(y) = \bigoplus_{y \in Y} \beta(y)$.
    \end{itemize}
\end{restatable}

Thus, if a strict statement is non-trivial, its ``non-strict version'' cannot be a contradiction. Further,  it can only be a tautology, 
if there is a variable set that is indifferent. 

The discussion above together with \Cref{pr:P134-sufficient-for-existence,pr:deductables-sat-P3-and-P4} allows us now to specify sets of non-trivial strict and non-strict statements that are sufficient to check w.r.t. P3 and P4 to guarantee the existence of a middle ground.
\begin{corollary}\label{thm:existence}
    Let $\Phi_1,\ldots,\Phi_n$ be non-trivial sets of statements.
    There exists a middle ground that includes a strict or a non-strict statement if and only if one of the following statements satisfies P3 and P4:
    \begin{align*}
        & \{\alpha \geq \beta \mid \alpha,\beta \in \und{V} \text{ s.th. } \alpha \geq \beta \text{ non-trivial}\} \\
        \cup & \{\alpha > \beta \mid \alpha,\beta \in \und{V} \text{ s.th. } \alpha > \beta \text{ non-trivial}, \alpha \geq \beta \text{ trivial}\}.
    \end{align*} 
\end{corollary}

    While 
    we can exactly determine the sets of statements in \Cref{thm:existence} and they are finite, they are   exponentially large on the number of variables. 
    As we see next, checking the postulates of the definition of a middle ground is {not in P
    for hierarchical models, unless P=NP and P=coNP.} 
    Thus, 
    we consider lexicographic models, which are a special case of hierarchical models. 
    {For these we can narrow down the set of preference statements that need to be checked for satisfying P3 and P4 to only $2\cdot|V|$ statements. Since checking the satisfaction of P3 and P4 is polynomial for lexicographic models~\citep{WGB15},  the existence of a middle ground is decidable in polynomial time.}

\subsubsection{Hierarchical Models}\label{sec:hier-algo}
{To analyse the complexity of deciding the existence of a middle ground for hierarchical models, we first show that it is NP-complete to decide consistency for hierarchical models. }

\citeauthor{WGB15}~(\citeyear{WGB15}) show that deciding $\Gamma \models \alpha \geq \beta$ is coNP-complete for their definition of hierarchical models which they call HCLP models, even if $\Gamma$ is a set of non-strict statements~\citep{WGB15}. Consequentially, deciding consistency of a set of statements $\Gamma$ is NP-complete under HCLP models. However, as outlined before, HCLP models are slightly differently defined than hierarchical models and the construction of the reduction from 3SAT to prove their central result is not transferable to hierarchical models. In particular their Lemma~2 does not hold if variable sets in models are allowed to be non-disjoint. 

To show NP-completeness of deciding consistency of a set of statements w.r.t. our definition of hierarchical models, we instead use a reduction from the Subset Sum Problem.\footnote{The same proof can also be used to show NP-completeness for HCLP models with an addition for the (trivial) case of the empty HCLP model. It thus offers a more concise alternative to the proof of \citeauthor{WG17}~(\citeyear{WG17}). More generally, it shows that deciding consistency is NP-complete even for models containing only one set of variables, and  only three preference statements.}
An instance of Subset Sum consists of a multi-set of integers $\mathcal{S}$ and a target integer $T$. The task is to decide whether there exists a multi-set $A \subseteq \mathcal{S}$ such that the sum of its element is $T$, i.e., $\sum_{a\in A} a = T$.
This problem is NP-complete even if all integers in $\mathcal{S}$ are positive~\citep{kleinberg2006algorithm}.

\begin{theorem}\label{thm:NPcompleteness-hierarchical-consisitency}
    Deciding consistency of a set of   preference statements is NP-complete w.r.t. hierarchical models with operators $\oplus$ that can be computed in time polynomial in the number of variables.
\end{theorem}
\begin{proof}
    To see that the consistency problem is in NP, we show that one can check in time polynomial in the number of variables and statements, whether a given hierarchical model $\pi=(Y_1, \ldots, Y_k)$ satisfies a set of given preference statements $\Gamma\subseteq \calL$. That is, for every non-strict statement $(\alpha \succeq \beta) \in \Gamma$ we need to check whether $\bigoplus_{y\in Y_i}\alpha(y) \geq \bigoplus_{y\in Y_i}\beta(y)$ for all $i=1, \dots, k$, and, for every strict statement $(\alpha \succ \beta) \in \Gamma$, we need to additionally check whether there exists $i\in \{1, \dots, k\}$ with $\bigoplus_{y\in Y_i}\alpha(y) \geq \bigoplus_{y\in Y_i}\beta(y)$. 
    By our assumption on $\oplus$ this can be computed in polynomial time.

    We show the completeness of the problem by a reduction from Subset Sum with positive integers.
    For this, let $\mathcal{S}$ be a multiset of positive integers and $T \in \mathds{N}$ a target.
    We construct three preference statements that, when satisfied together, force a hierarchical model to contain a variable set that corresponds to a solution of Subset Sum for $\mathcal{S}$ and $T$.

    \noindent \textbf{Variables:} We construct a variable $v_a$ for each $a \in \mathcal{S}$ and one variable $v_T$. Denote the set of variables by $V$. Then $|V| = |\mathcal{S}|+1$. Let the domain of each variable be $\mathds{N}$. 

    \noindent \textbf{Operator:} The operator $\oplus$ is the normal addition.

    \noindent \textbf{Preference Statements:} Consider preference statements 
    \[\Phi = \{(\alpha_T > \beta_T), (\alpha_\Sigma \geq \beta_\Sigma), (\beta_\Sigma \geq \alpha_\Sigma)\}\] 
    for alternatives $\alpha_T, \beta_T, \alpha_\Sigma, \beta_\Sigma$ that are defined as follows:
    \begin{align*}
        & \alpha_T(v_c) = \begin{cases}
        0 \text{ if } c\in\mathcal{S}\\
        1 \text{ if } c = T
    \end{cases} &&\beta_T(v) = 0 \; \forall v \in V \\
    \text{ and } \quad
        & \alpha_\Sigma(v_c) = \begin{cases}
        c \text{ if } c\in\mathcal{S}\\
        0 \text{ if } c = T
    \end{cases} 
    &&\beta_\Sigma(v_c) = \begin{cases}
        0 \text{ if } c\in\mathcal{S}\\
        T \text{ if } c = T.
    \end{cases}
    \end{align*}

    \noindent \textbf{Satisfaction:} 
    We first show that if there exists a hierarchical model satisfying $\Phi$ then there exists a Subset Sum solution. Then we show the reverse.

    Assume that there is a hierarchical model $\pi$ with $\pi \models \Phi$. Because $\alpha_T > \beta_T$ is a strict statement, but all variables are indifferent under $\alpha_T$ and $\beta_T$ except $v_T$, $\pi$ must contain $v_T$ 
    in some variable set. Let $C$ be the first such variable set in $\pi$ with $v_T \in C$. Then by $\pi \models (\alpha_\Sigma \geq \beta_\Sigma)$, either (1) $C$ is preceded by another variable set $C'$ in $\pi$, or (2) $C$ contains other variables and $\sum_{v_c \in C} \alpha_\Sigma(v_c) \geq \sum_{v_c \in C} \beta_\Sigma(v_c) = T$.
    By $\pi \models (\beta_\Sigma \geq \alpha_\Sigma)$ and because integers in $\mathcal{S}$ are positive, case (1) is not possible. 
    Thus assume that case (2) holds and let $A = \{a \in \mathcal{S} \mid v_a \in C \setminus \{v_T\}\}$. Then $\sum_{a \in A} a = \sum_{v_c \in C \setminus \{v_T\}} \alpha_\Sigma(v_c) \geq T$.
    Further, by $\pi \models (\beta_\Sigma \geq \alpha_\Sigma)$ and because there is no other variable set preceding $C$ in $\pi$, we have $T = \sum_{v_c \in C} \beta_\Sigma(v_c) \geq \sum_{v_c \in C} \alpha_\Sigma(v_c) = \sum_{a \in A} a$.
    So, if $\pi \models \Phi$ then $\pi$ contains a set of variables $C$ corresponding to $T$ and integers $A \subseteq \mathcal{S}$ such that $T = \sum_{a \in A} a$.

    For the reverse, assume there exists a multiset $A$ that is a subset 
    of integers $\mathcal{S}$  with  $T = \sum_{a \in A} a$.
    Then, as shown before, the hierarchical model $\pi = (\{v_a \mid a\in A\}\cup\{v_T\})$ satisfies $\Phi$.
    Thus there exists a hierarchical model satisfying $\Phi$ iff there exists a subset of $\mathcal{S}$ that sums to $T$.
    Because construction of $\Phi$ is polynomial in the size of the Subset Sum instance, and Subset Sum is NP-complete, so is deciding consistency for hierarchical models and   preference statements.
\end{proof}

Recall that 
for preference statements $\Gamma, \alpha > \beta$, we have $\Gamma \models (\alpha > \beta)$ if and only if $\Gamma \cup \{\alpha \leq \beta\}$ is inconsistent \citep{WG17}. Similarly, $\Gamma \models (\alpha \geq \beta)$ if and only if $\Gamma \cup (\alpha < \beta)$ is inconsistent.
Then \Cref{thm:NPcompleteness-hierarchical-consisitency} has the following consequences for the complexity of deduction and testing middle ground. 
\begin{corollary}
    Deciding $\Gamma \models \varphi$ for preference statements $\Gamma$ and $ \varphi$ w.r.t. hierarchical models is coNP-complete.
\end{corollary}
\begin{corollary}
    Let $\Phi_1,\ldots,\Phi_n$ be non-trivial sets of statements.
    Let $\Phi$ be a given set of statements and consider the satisfaction relation  w.r.t. hierarchical models. Then,
    deciding whether $\Phi$ satisfies (P1) is NP-complete and
    deciding whether $\Phi$ satisfies (P4) for $\Phi_1,\ldots,\Phi_n$ is coNP-complete.
\end{corollary}

\subsubsection{Lexicographic Models}\label{sec:lex-algo}
For lexicographic models, we are able to decrease the set of statements in \Cref{thm:existence} that need to be considered to determine existence of a middle ground to a polynomial size. 
We first show that it is sufficient to consider binary variables.
Informally, this follows from the fact that lexicographic models only consider ordinal relations between the variable assignments and not their actual values.

\begin{restatable}{proposition}{binary}\label{prop:binary}
    Given statement $\phi=(\alpha\geq \beta)$, let $\phi^b=(\alpha^b\geq \beta^b)$ be the result of replacing $\alpha(v)$ and $\beta(v)$ in $\phi$ by 
    \begin{itemize}
        \item  $1$ and $0$, respectively, if $\alpha(v)>\beta(v)$;
        \item  $0$ and $1$, respectively, if $\alpha(v)<\beta(v)$; and
        \item  $0$ and $0$, respectively, if $\alpha(v)=\beta(v)$, for all $v\in V$.
    \end{itemize}
Under the assumption that $1 > 0$, we have for all lexicographic models $\pi$, that $\pi\models\phi$ iff $\pi\models\phi^b$.
\end{restatable}

We can further decrease the set of statements to consider for testing existence of a middle ground as follows.

\begin{theorem}\label{thm:existence-lex}
Consider inference based on lexicographic models on variables $V$. 
    Let $\Vec{0}$ denote the vector of $|V|$ zeros and let $\Vec{0}_v$ be the same as $\Vec{0}$ but with $1$ at the position corresponding to a variable $v \in V$.
    Similarly, let $\Vec{1}$ denote the vector of $|V|$ ones and let $\Vec{1}_v = \Vec{1} - \Vec{0}_v$.
    There exists a middle ground for non-trivial sets of statements $\Phi_1,\ldots,\Phi_n$ if and only if one of the following statements satisfy P3 and P4:
    \[\{\Vec{1}_v \geq \Vec{0}_v \mid v \in V\} 
        \cup \{\Vec{1}_v > \Vec{0} \mid v \in V\} .\]
        
\end{theorem}
\begin{proof}
    Suppose there exists a middle ground for $\Phi_1,\ldots,\Phi_n$ that includes a non-trivial statement $\phi$.
    By Proposition~\ref{prop:binary}, and under the assumption that $1 > 0$ we have for all lexicographic models $\pi$, that $\pi\models\phi$ iff $\pi\models\phi^b$. Here $\phi^b$ is the statement over binary variables as defined in Proposition~\ref{prop:binary}.
    We can thus focus our following argumentation on $\phi^b$ instead of $\phi$.

    Assume that $\phi^b$ is a non-strict statement $\alpha \geq \beta$. 
    Then because it is non-trivial, and in particular not a tautology, there exists 
    variable $v$ 
    such that $\alpha(v)<\beta(v)$. 
    Any lexicographic
    model that satisfies $\alpha \geq \beta$ must include another variable $v'$ preceding $v$ 
    or not include $v$ at all.
    Thus, any such model also satisfies $\vec{1}_v \geq \vec{0}_v$, i.e., $(\alpha \geq \beta) \models (\vec{1}_v \geq \vec{0}_v)$.

    Now assume that $\phi^b$ is a strict statement $\alpha > \beta$. Then because it is non-trivial, by \Cref{le:non-strict-version-of-strict-non-trivial}, it entails either (1) a non-trivial non-strict statement or (2) there exists a variable $v$ such that $\alpha(v)=\beta(v)(=0)$ and $\alpha(v')\geq\beta(v')$ for all   variables $v'\in V\setminus \{v\}$. 
    In case (1), by our arguments above, $\alpha > \beta$ also entails a non-trivial non-strict statement $\vec{1}_v \geq \vec{0}_v$ for some $v \in V$.
    For case (2), we show that $\alpha > \beta$ entails the statement $\vec{1}_v > \vec{0}$.
    Because $\alpha > \beta$ is a strict statement and because $\alpha(v)=\beta(v)$ the lexicographic model that \textit{only} includes $v$ does not satisfy $\alpha > \beta$. Thus, any lexicographic model satisfying $\alpha > \beta$ must also include some other variable $v'$. Any such model then also satisfies $\vec{1}_v > \vec{0}$.

    As shown above, any middle ground of statements entails a statement in 
    $\{\Vec{1}_v \geq \Vec{0}_v \mid v \in V\} \cup \{\Vec{1}_v > \Vec{0} \mid v \in V\}$.
    By \Cref{pr:deductables-sat-P3-and-P4}, this statement satisfies P3 and P4.

    For the converse direction, assume there exists a statement in $\{\Vec{1}_v \geq \Vec{0}_v \mid v \in V\} 
        \cup \{\Vec{1}_v > \Vec{0} \mid v \in V\}$
        that satisfies P3 and P4.
    Because all such statements are by construction non-trivial,
    they also satisfy P1.
    Then by \Cref{pr:P134-sufficient-for-existence} there exists a middle ground.
\end{proof}
\citeauthor{WGB15}~(\citeyear{WGB15}) establish that checking consistency or inference is polynomial-time solvable for lexicographic models and strict and non-strict preference statements. In consequence, we can check in polynomial-time whether a statement satisfies P3 or P4. By \Cref{thm:existence-lex}, we only need to do this for $2\cdot|V|$ many statements for lexicographic models. 

\begin{corollary}
    Let $\Phi_1,\ldots,\Phi_n$ be non-trivial sets of statements.
    Checking whether there exists a middle ground w.r.t. lexicographic models that includes a (non-trivial) strict or non-strict statement is polynomial-time solvable.
\end{corollary}

We summarise the algorithm to decide existence of a middle ground in Algorithm~\ref{algo:lex-basic-middle-ground-existence}. 
Here, the set of statements $\mathcal{L}_{b}$ has cardinality $2\cdot|V|$. To construct $\Psi_{3}$ we need to perform $2\cdot |V|\cdot|\bigcup^n_{i=1} \Phi_{i}|$ consistency checks between two statements.
To construct $\Psi_{4}$ we need to perform $2\cdot |V|\cdot n$ checks to see if a statement can be deduced from a stakeholder's statements.

\begin{algorithm}[h]
\SetKwInOut{Input}{Input}\SetKwInOut{Output}{Output}
\Input{Sets of non-trivial preferences $\Phi_1,\ldots,\Phi_n$.}
\Output{`yes' if it
exists, `no' otherwise.}
\BlankLine
\lIf{$\bigcup^n_{i=1} \Phi_{i}$ is consistent}{
\textbf{return} $\bigcup^n_{i=1} \Phi_{i}$
}
$\mathcal{L}_{b} := \{(\Vec{1}_v \geq \Vec{0}_v), (\Vec{1}_v > \Vec{0}) \mid v \in V\}$\;
$\Psi_{3} := \{\varphi \in \mathcal{L}_{b} \mid \forall \psi \in \bigcup^n_{i=1} \Phi_{i} \colon \{\varphi,\psi\}\text{ consistent}\}$\;
$\Psi_{4} := \{\varphi \in \mathcal{L}_{b} \mid \exists i \in \{1, \dots, n\} \colon \Phi_{i} \models \varphi\}$\;
\lIf{$\Psi_{3} \cap \Psi_{4} = \emptyset$}
{\textbf{return} `no'}
\lElse{\textbf{return} `yes'}
\caption{Existence of Middle Ground 
}\label{algo:lex-basic-middle-ground-existence}
\end{algorithm}

One can also employ Algorithm~\ref{algo:construct-middle-ground} to \textit{construct} a middle ground for lexicographic models. In case the stakeholders statements are inconsistent (which can be checked in polynomial time), by Proposition~\ref{prop:binary}, one can then focus on binary statements instead of the complete language $\calL$.
The number of strict and non-strict statements over $|V|$ variables with binary domains is $O(2^{2|V|-1})$.
We then need to check every one of the $O(2^{2|V|-1})$ statements on whether they satisfy P1, P3 and P4. 
Thus, 
while each of these tests individually can be done in polynomial time,
Algorithm~\ref{algo:construct-middle-ground} remains exponential even for lexicographic models. It remains an open question whether there exist tractable algorithms to solve this problem.

\section{Conclusion}\label{sec: conclusion}
We investigate the notion of 
middle ground, exploring its properties, including the fact that it may not exist or be unique. We establish necessary conditions for its existence and describe an algorithmic procedure for both existence checking and construction. Our case study focuses on preference statements, with a notable application in moral preferences and hiring: while the general problem is coNP-complete, we show that deciding existence is tractable for lexicographic models. 

Our postulate P3 concerns statements in the middle ground individually and not the whole set. It remains open whether a stronger version of this or other postulates lead to more tractable algorithms or a unique middle ground.
Other
future work may analyse the tractability of constructing middle grounds for lexicographic models and explore the concept for 
non-preference-based languages, such as propositional logic.

\section*{Acknowledgments}    
This work was supported by the Research Council of Norway through its Centre of Excellence: Integreat - The Norwegian Centre for knowledge-driven machine learning, project number 332645.

\bibliographystyle{named}
\bibliography{ref}

\newpage

\appendix
\section{Discussion on Middle Ground Definition}\label{sec:mgdef} 
\newcommand{\formula}[1]{\ensuremath{\mathcal{F}_{#1}}\xspace}
\newcommand{\BK}{\Bmc}
\newcommand{\ant}[1]{{\sf ant}(#1)}
\newcommand{\cons}[1]{{\sf con}(#1)}
In this section, we discuss the similarities and differences between the definition of middle ground by \citeauthor{DBLP:journals/aamas/OzakiRS24}~(\citeyear{DBLP:journals/aamas/OzakiRS24}) and ours. For the convenience of the reader, we recall here the relevant notions. 
The work by  \citeauthor{DBLP:journals/aamas/OzakiRS24}~(\citeyear{DBLP:journals/aamas/OzakiRS24})  uses Horn expressions and the notion of \emph{coherence}. From now on, we may refer to their definition of middle ground as \emph{middle ground for Horn expressions}.
\newcommand{\Pmc}{\ensuremath{\mathcal{P}}\xspace}
\begin{definition}[Horn Expressions]
Let \Pmc be a set of propositional variables.
    A \emph{Horn clause} is a disjunction of literals with at most one positive literal, where a literal is a variable in $\Pmc$ or its negation. A Horn clause is called \emph{definite} if it has exactly one positive literal. A \emph{definite Horn expression} is a set of definite Horn clauses. In symbols, a Horn clause $\phi$  is of the form $(\neg p_1\vee \ldots\vee \neg p_n\vee q)$. 
Such a Horn clause $\phi$ can be equivalently written as a \emph{rule} of the form $(p_1\wedge \ldots \wedge p_n )\rightarrow q$, where the \emph{antecedent} of $\phi$, written $\ant{\phi}$, is $\{p_1, \ldots, p_n\}$, and the \emph{consequent},  $\cons{\phi}$, is the variable $q$. 
\end{definition}
It is assumed that stakeholders share a \emph{background knowledge}
which describes propositions that cannot be both true, e.g., a person cannot be both an adult and a teenager. 
Given a propositional variable $q$, the authors
write $\overline{q}$ for the set of variables that cannot be true when $q$ is true. E.g., if \Pmc has propositional variables for adult, teenager, and child then 
$\overline{\mathsf{child}}=\{\mathsf{adult, teenager}\}$. 
This background knowledge is used to guide the construction of a middle ground for Horn expressions.
  We are now ready to introduce the notions of coherence and middle ground for Horn expressions.
\begin{definition}[Coherence]
A definite Horn clause $\phi$ is  { \bf coherent  } with
a  
Horn expression ${\formula{}}$
if $\formula{}\setminus\{\phi\}\not\models \phi$ and 
\begin{itemize}
\item there is no $\psi\in\formula{}$ such that 
$\psi\Rightarrow_{\formula{}}\phi$ or $\phi\Rightarrow_{\formula{}}\psi$ while
$\cons{\psi}\in \overline{\cons{\phi}}$ (note that $\cons{\psi}\in \overline{\cons{\phi}}$    implies 
$\cons{\phi}\in \overline{\cons{\psi}}$).
\end{itemize}
The set \formula{} is {\bf{coherent}} if all 
  $\phi\in\formula{}$ are coherent with $\formula{}$ (and \emph{incoherent} otherwise).
  \end{definition}

\begin{definition}[Middle Ground {for Horn Formulas}]\label{def:consensus_appendix}
Let $\formula{1},\ldots \formula{n}$   be  definite Horn expressions, each associated with a stakeholder 
$i\in\{1,\ldots,n\}$.
Let  $\BK$ be a set describing  
background knowledge. A formula \formula{} is a {\bf{middle ground}} for $\formula{1},\ldots, \formula{n}$ and \BK
if it 
satisfies   the following postulates: 
\begin{itemize}[leftmargin=7.5mm]
\item[(P1)]  \formula{} is coherent; 
\item[(P2)] if  $\bigcup^n_{i=1} \formula{i}$ is coherent, then $\formula{}  \equiv \bigcup^n_{i=1} \formula{i}$; 
\item[(P3)] for all $i \in\{1,\ldots,n\}$ and all $\phi \in \formula{i}$, we have that 
 $\formula{}\not\models \ant{\phi} \rightarrow p$ with $p\in\overline{\cons{\phi}}$; 
\item[(P4)] for each $\phi\in\formula{}$, there is $\psi\in\bigcup^n_{i=1} \formula{i}$ with $\{\psi\}\models\phi$;
\item[(P5)] {for all $i \in\{1,\ldots,n\}$ and all $\phi \in \formula{i}$,
there is $\psi\in\formula{}$ that is not a tautology   such that $\{\phi\}\models\psi$; and}
\item[(P6)] for all  $\phi\in\formula{}$, if there is $p\in \ant{\phi}$ such that, for all $q\in\overline{p}$, 
$\formula{}\cup\{\phi^{q\setminus p}\}$ is coherent {and
there is $i\in\{1,\ldots,n\}$ such that $\formula{i}\models\phi^{q\setminus p}$
then $\formula{i}\not\models\phi^{-p}$.}
\end{itemize}
\end{definition}



The counterpart of the notion of being `coherent' by \citeauthor{DBLP:journals/aamas/OzakiRS24}~(\citeyear{DBLP:journals/aamas/OzakiRS24}) in our work is given by the notion of being `non-trivial'. The crucial reason for this difference is because, since the statements of the stakeholders are definite Horn expressions, the union of such statements  is always consistent: any set definite Horn expressions is satisfied by the interpretation that sets all variables to true. So, there is no point in analysing inconsistencies there. 
Here, since we do not impose any specific format in the language of a satisfaction system, even if we require stakeholders to be individually consistent, the union of their statements could be inconsistent.

Considering the notion of being `coherent' as a counterpart to the notion of being `non-trivial', we have that P1 coincide in both definitions.
The same happens for P2: since $\Phi_i$ is non-trivial, it is falsifiable.
By  \Cref{def:satisfactionsystem} (see Point~1), $\mathsf{mod}(\bigcup^n_{i=1} \Phi_{i})\subseteq \mathsf{mod}(\{\Phi_i\})$, so $\bigcup^n_{i=1} \Phi_{i}$ is falsifiable.  
This means that ``consistent'' in our P2 definition could be replaced by `non-trivial' without causing any change in the definition of middle ground. 
So, P2 also coincides in both definitions.

Regarding P3, the intuition of P3 for Horn expressions is that the middle ground does not include a Horn clause that is in ``direct opposition'' with a statement of a stakeholder. This idea is expressed using the syntax of Horn expressions, in particular, the antecedent and the consequent of a Horn clause. Since we provide a definition for a generic satisfaction system, our approximation of this idea is to say that, for all statements in the middle ground $\phi$ and for all statements of the stakeholders $\phi_i$, there is a model that satisfies $\phi$ and $\phi_i$. 
Our P4 is less restrictive than P4 for Horn expressions. That is, if P4 holds in the middle ground for Horn expressions then it holds for our notion of middle ground.

Regarding P5, in both definitions, the goal is to the retain  the statements of the stakeholders, even if in a weaker form. 
There is no counterpart of P6 for Horn expressions in our work. P6 complements P5 in that work as a way of maximizing the retention of the statements of the stakeholders, while denying Horn clauses deemed irrelevant (since they use symbols unrelated to incoherences between stakeholders)~\cite{DBLP:journals/aamas/OzakiRS24}. This specific notion is not something we could mimic in our generic definition, using satisfaction systems. We also did not find any condition that would approximate P6 in a sensible way for our study case on preference statements.

A natural question is whether the instantiation of our definition to the satisfaction system of Horn expressions would yield the same result. The answer is `no'. As mentioned earlier, there is no way of creating a logical inconsistency with definite Horn expressions due to the lack of the negation operator. By our P2, the union would always be the middle ground.

\section{Proofs in Section~\ref{sec:mg}}
\setcounter{theorem}{0}
\ptwo*
\begin{proof}
    By assumption $\Phi_1,\ldots,\Phi_n$   are  
non-trivial, so they are consistent and falsifiable. If   $\Phi_1,\ldots,\Phi_n$  are falsifiable
then $\bigcup^n_{i=1} \Phi_{i}$ is also falsifiable (see Point~1 in~\Cref{def:satisfactionsystem} on monotonicity). By assumption   $\bigcup^n_{i=1} \Phi_{i}$ is consistent. So  $\bigcup^n_{i=1} \Phi_{i}$ is non-trivial P1. Also, $\bigcup^n_{i=1} \Phi_{i}$ trivially satisfies P2. 
Regarding P3,  by assumption $\bigcup^n_{i=1} \Phi_{i}$ is consistent. So, in particular, it is satisfiable. Then any model that satisfies $\bigcup^n_{i=1} \Phi_{i}$ will also satisfy each $\Phi_i$, with $1\leq i\leq n$. 
P4 also holds trivially by definition of $\bigcup^n_{i=1} \Phi_{i}$. 
Finally, $\bigcup^n_{i=1} \Phi_{i}$ satisfies P5 because P2 needs to be satisfied. 
\end{proof}

\setcounter{theorem}{1}
\ponepthreepfour*
\begin{proof}
{By \Cref{prop:postulatetwo},
if  $\bigcup^n_{i=1} \Phi_{i}$ is consistent then we are done. 
Otherwise, if  $\bigcup^n_{i=1} \Phi_{i}$ is not consistent,
then we satisfy (P2), so we do not need to take it into account in the following argument.
   If a non-trivial set of statements $\Phi$ satisfies (P1), (P3), and (P4) then either it also satisfies (P5) and we are done or there is another non-equivalent and non-trivial set of statements that satisfies (P1), (P3) and (P4) and entails $\Phi$.
    Since there are finitely many non-equivalent sets of statements to consider (there are finitely many variables and the domain is finite for all variables), a non-trivial set of statements that satisfies (P3), (P4), and (P5) exists. }
\end{proof}

\deductables*
\begin{proof}
    Consider a statement $\varphi$ such that $\phi \models \varphi$ for some $\phi \in \Phi$, where $\Phi$ satisfies P3 and P4 w.r.t. stakeholders $\Phi_1,\ldots,\Phi_n$.
    
    By P3, for all $i \in\{1,\ldots,n\}$ and all $\phi_i\in\Phi_i$, 
there exists a model $\pi$ such that
$\pi\models\Phi$ and
$\pi\models\phi_i$. Because $\pi \models \Phi$ implies $\pi \models \phi$ and thus $\pi \models \varphi$, $\varphi$ also satisfies P3.

Because $\Phi$ satisfies P4, there exists $i\in\{1,\ldots,n\}$ such that 
$\Phi_i\models\phi$. Thus, there exists $i\in\{1,\ldots,n\}$ such that 
$\Phi_i\models\varphi$.
Hence, $\varphi$ satisfies P3 and P4.   
\end{proof}

\equivorinconsistent*
 \begin{proof}
    First, consider the case that $\bigcup^n_{i=1} \Phi_{i}$ is consistent. Then by P2, $\Phi_1  \equiv \bigcup^n_{i=1} \Phi_{i}$ and $\Phi_2 \equiv \bigcup^n_{i=1} \Phi_{i}$. Thus, $\Phi_1 \equiv \Phi_2$. For the remainder of the proof, we concentrate on the case that $\bigcup^n_{i=1} \Phi_{i}$ is not consistent.
    
     Because of P1, all statements in $\Phi_1 \cup \Phi_2$ are falsifiable. Thus, $\Phi_1 \cup \Phi_2$ is falsifiable.

    Now assume that $\Phi_1 \cup \Phi_2$ is consistent. Then $\Phi_1 \cup \Phi_2$ satisfies P1, and by our previous assumption also P2. 
    Further, because $\Phi_1$ and $\Phi_2$ satisfy P3, we have: For every statement  $\phi\in\Phi_1 \cup \Phi_2$ and for all $i \in\{1,\ldots,n\}$ and all $\phi_i\in\Phi_i$,  
there is $\pi$ such that
$\pi\models\phi$ and
$\pi\models\phi_i$. Thus $\Phi_1 \cup \Phi_2$ satisfies P3.
Similarly, because $\Phi_1$ and $\Phi_2$ satisfy P4, we have: there is $i\in\{1,\ldots,n\}$ such that 
$\Phi_i\models\phi$. Thus $\Phi_1 \cup \Phi_2$ satisfies P4.

    Since $\Phi_1$ and $\Phi_2$ are middle grounds, they also satisfy P5. We have shown that by our assumptions $\Phi_1 \cup \Phi_2$ satisfies P1-P4. Further, $\Phi_1 \cup \Phi_2 \models \Phi_1$ and $\Phi_1 \cup \Phi_2 \models \Phi_2$, respectively.
    Thus, by P5, $\Phi_1 \models \Phi_1 \cup \Phi_2$ and $\Phi_2 \models \Phi_1 \cup \Phi_2$, respectively. Hence $\Phi_1 \equiv \Phi_1 \cup \Phi_2 \equiv \Phi_2$.

    This shows that either $\Phi_1 \equiv \Phi_2$ or $\Phi_1 \cup \Phi_2$ is inconsistent.
 \end{proof}

\section{Proofs in Section~\ref{sec:mgpreferences}}

\setcounter{theorem}{6}
We characterise tautologies (i.e., statements satisfied by every model) and contradictions (i.e., statements satisfied by no model) 
as follows.
\begin{restatable}
{lemma}{tautologycontradiction}\label{le:tautology-and-contradiction}
   Let $\alpha, \beta \in \outc$ be two alternatives.  
   \begin{itemize}
       \item \hspace{.75cm} $\alpha \geq \beta$ is a tautology \\
      $\Leftrightarrow\quad$ $\bigoplus_{y \in Y} \alpha(y) \geq \bigoplus_{y \in Y} \beta(y)$ for all $Y \subseteq V$\\
       $\Leftrightarrow\quad$ $\beta > \alpha$ is a contradiction.
       \item \hspace{.75cm} $\alpha > \beta$ is a tautology\\
       $\Leftrightarrow\quad$ $\bigoplus_{y \in Y} \alpha(y) > \bigoplus_{y \in Y} \beta(y)$ for all $Y \subseteq V$\\
       $\Leftrightarrow\quad$ $\beta \geq \alpha$ is a contradiction.
   \end{itemize}
\end{restatable}

    \begin{proof}
Let $\alpha, \beta \in \outc$.
    We first show that $\alpha \geq \beta$ is a tautology if and only if $\bigoplus_{y \in Y} \alpha(y) \geq \bigoplus_{y \in Y} \beta(y)$ for all $Y \subseteq V$. If $\alpha \geq \beta$ is a tautology, then every model satisfies $\alpha \geq \beta$. In particular, all models $(Y)$ that have only one level of variables $Y$ satisfy the statement. Hence $\bigoplus_{y \in Y} \alpha(y) \geq \bigoplus_{y \in Y} \beta(y)$ for all $Y \subseteq V$. Reversely, if $\bigoplus_{y \in Y} \alpha(y) \geq \bigoplus_{y \in Y} \beta(y)$ for all $Y \subseteq V$, then there exists no set opposing the statement and the statement must be true for any model.
    Analogously, we can show that $\alpha > \beta$ is a tautology if and only if $\bigoplus_{y \in Y} \alpha(y) > \bigoplus_{y \in Y} \beta(y)$ for all $Y \subseteq V$.

    {Now consider a preference statement $\varphi$. 
    As mentioned, because $\succeq _\pi$ is a total pre-order over the alternatives, $\pi \not\models \alpha \geq \beta$ iff $\pi \models \beta > \alpha$. In other words, either $\pi$ satisfies 
    a preference statement $\alpha > \beta$ or its ``complement'' $\beta \geq \alpha$.
    Thus $\varphi$ is a contradiction (not satisfied by any model) iff its complement is a tautology (satisfied by every model).
    } 
\end{proof}
\nonunique*
\begin{proof}
     Consider the   alternatives defined over four binary variables in \Cref{tab:my_label}. For simplicity, we assume that the value of any sum of variables is the same for all alternatives and omit such values in \Cref{tab:my_label}. We also omit   alternatives that are not explicitly needed for our argument.
     \begin{table}[h]
         \centering  
        \begin{tabular}{c c c c c}
                        & $X$ & $Y$ & $Z$ & $W$\\ \hline
           $\alpha =$  & 1 & 0 & 0 & 0\\
           $\beta =$  & 0 & 1 & 0 & 0 \\
            $\alpha' =$  & 0 & 0 & 1 & 0\\
           $\beta' =$  & 0 & 0 & 0 & 1 \\
           $\gamma =$  & 1 & 0 & 1 & 0 \\
           $\delta =$  & 0 & 1 & 0 & 1 
        \end{tabular}
         \caption{Alternatives.}
         \label{tab:my_label}
     \end{table}
     
         Consider two stakeholders expressing the following non-trivial sets of preference statements:
        \[\Phi_1 = \{(\alpha > \beta),(\alpha' > \beta')\}, \quad \Phi_2 = \{ (\beta > \alpha), (\beta' > \alpha')\}.\]
            The stakeholder's statements are consistent individually, but inconsistent together. This means that the union of $\Phi_1$ and $\Phi_2$ cannot be a middle ground for the preferences of these stakeholders. We want to show that there are at least two non-equivalent middle grounds for  $\Phi_1$ and $\Phi_2$.
            Now, consider the following preference statements:
            \[\psi_1 = (\gamma > \delta), \quad \psi_2 =  (\delta > \gamma).\]
        We have that 
        \begin{enumerate}
            \item $\psi_1$ and $\psi_2$ are individually non-trivial;
            \item $\psi_1$ and $\psi_2$ are  inconsistent together; 
            \item for all $i,j\in \{1,2\}$ and all $\phi_i\in\Phi_i$, there is $\pi$ such that
            $\pi\models \psi_j$ and $\pi\models \phi_i$;
            \item for $i\in \{1,2\}$, $\Phi_i\models\psi_i$.
        \end{enumerate}

\noindent
\textit{Proof of Point (1)} We have that e.g.
$(\{X\})$ satisfies 
$\psi_1 = (\gamma > \delta)$ and
$(\{Y\})$ does not satisfy $\psi_1$. So $\psi_1$ is non-trivial. Also, we have that e.g.
$(\{Y\})$ satisfies 
$\psi_2 = (\delta > \gamma)$ and
$(\{X\})$ does not satisfy $\psi_2$. So $\psi_2$ is non-trivial. 

\smallskip
\noindent
\textit{Proof of Point (2)} Any model that satisfies $\psi_1$ needs to have \emph{at least one of} $\{X\}$ and 
$\{Z\}$ before \emph{both} $\{Y\}$ and 
$\{W\}$, provided $\{Y\}$ and/or 
$\{W\}$ occur in the model (recall that, for succinctness, all sums of variables have the same value for all alternatives). For $\psi_2$, it is the other way round. Any model that satisfies $\psi_2$ needs to have \emph{at least one of} $\{Y\}$ and 
$\{W\}$ before \emph{both} $\{X\}$ and 
$\{Z\}$,  provided $\{X\}$ and/or 
$\{Z\}$ occur in the model. So, there is no model that satisfies both $\psi_1$ and $\psi_2$. In other words, $\psi_1$ and $\psi_2$ are  inconsistent together.

\smallskip
\noindent
\textit{Proof of Point (3)}
By Point (4) every model that satisfies
$\Phi_1$ also satisfies $\psi_1$ (and since $\Phi_1$ is non-trivial at least one such model exists). We now consider $\Phi_2$ and $\psi_1$.
Let $\phi_2=(\beta > \alpha)$ and 
$\phi'_2=(\beta' > \alpha')$. 
We have that  $(\{Z\},\{Y\})$ satisfies both
$\phi_2$ and $\psi_1$. Also,
$(\{X\},\{W\})$ satisfies both
$\phi'_2$ and $\psi_1$.
Now, by Point (4) every model that satisfies
$\Phi_2$ also satisfies $\psi_2$ (and at least one such model exists). So we need to consider $\Phi_1$ and $\psi_2$.
Let $\phi_1=(\alpha > \beta)$ and 
$\phi'_1=(\alpha' > \beta')$. 
We have that  $(\{W \},\{ X\})$ satisfies both
$\phi_1$ and $\psi_2$. Also,
$(\{Y\},\{ Z\})$ satisfies both
$\phi'_1$ and $\psi_2$.

\smallskip
\noindent
\textit{Proof of Point (4)} 
We argue that
every model that violates $\psi_1$ also violates  $\Phi_1$ (note that the converse does not hold, e.g. $(\{Z\})$ violates $\Phi_1$ but it does not violate $\psi_1$). To violate $\psi_1$ we need 
that 
\begin{itemize}
    \item $\{Y\}$ occurs in the model and it is not after  $\{X\}$ or $\{Z\}$; or
    \item $\{W\}$ occurs in the model and it is not after  $\{X\}$ or $\{Z\}$.
\end{itemize}
In the former case, $\{Y\}$ occurs in the model and it is not after  $\{X\}$ or $\{Z\}$. Then, in particular, it is not after $\{X\}$.
So it violates $\phi_1$. In the latter, $\{W\}$ occurs in the model and it is not after  $\{X\}$ or $\{Z\}$. Then, in particular, it is not after $\{Z\}$.
So it violates $\phi'_1$.
The proof that
every model that violates $\psi_2$ also violates  $\Phi_2$ is analogous. To violate $\psi_2$ we need 
that 
\begin{itemize}
    \item $\{X\}$ occurs in the model and it is not after  $\{Y\}$ or $\{W\}$; or
    \item $\{Z\}$ occurs in the model and it is not after   $\{Y\}$ or $\{W\}$.
\end{itemize}
In the former case, $\{X\}$ occurs in the model and it is not after  $\{Y\}$ or $\{W\}$. Then, in particular, it is not after $\{Y\}$.
So it violates $\phi_2$. In the latter, $\{Z\}$ occurs in the model and it is not after   $\{Y\}$ or $\{W\}$. Then, in particular, it is not after $\{W\}$.
So it violates $\phi'_2$.

\smallskip        
        To conclude, we claim that
        there are at least two non-equivalent middle grounds, one that contains $\psi_1$ and another one that contains $\psi_2$.
        Indeed,
        note that Points (1), (3), (4) and \Cref{{thm:algmiddleground}} imply that that there is a middle ground for 
        $\Phi_1$ and $\Phi_2$ that contains
        $\psi_1$
        (plus possibly other statements, so as to satisfy P5 in  \Cref{def:consensus_appendix}) and there is a middle ground for 
        $\Phi_1$ and $\Phi_2$ that contains
        $\psi_2$. Point (2) implies that
        there is no middle ground that contains both $\psi_1$ and $\psi_2$ (otherwise P1 in  \Cref{def:consensus_appendix} would be violated). So there are at least two non-equivalent middle grounds for $\Phi_1$ and $\Phi_2$.      
\end{proof}

\setcounter{theorem}{10}
\nontrivial*
\begin{proof}
Let $\alpha > \beta$ be non-trivial.
Any statement is either a tautology, a contradiction or non-trivial. We show that $\alpha \geq \beta$ cannot be a contradiction and can only be a tautology under specific circumstances.

Assume $\alpha \geq \beta$ is a contradiction. Then, by \Cref{le:tautology-and-contradiction}, $\beta > \alpha$ is a tautology. But then also $\beta \geq \alpha$ is a tautology and thus $\alpha > \beta$ is a contradiction. This contradicts the assumption that $\alpha > \beta$ is non-trivial.

Consider the case that $\alpha \geq \beta$ is a tautology. Then, by \Cref{le:tautology-and-contradiction}, $\bigoplus_{y \in Y} \alpha(y) \geq \bigoplus_{y \in Y} \beta(y)$ for all $Y \subseteq V$. Because $\alpha > \beta$ is non-trivial and in particular not a tautology, there exits $Y \subseteq V$  such that $\bigoplus_{y \in Y} \alpha(y) = \bigoplus_{y \in Y} \beta(y)$.
\end{proof}

\setcounter{theorem}{15}
\binary*
\begin{proof} We have that 
    $\pi\models \alpha \geq \beta $, if and only if $\alpha \succeq _\pi \beta$. 
Since $\pi$ is a lexicographic model, we have
$\alpha \succeq _\pi \beta$ iff
\begin{enumerate}[label=(\roman*)]
    \item for all $i=1, \ldots, k$,  $ \alpha(y_i) =  \beta(y_i)$, or
    \item there exists $i \in \{1, \ldots, k\}$ such that $ \alpha(y_i) >  \beta(y_i)$ and for all $j < i$, $  \alpha(y_j) = \beta(y_j)$.
\end{enumerate}
By definition of $\alpha^b$ and $\beta^b$, we have $ \alpha(y_i) =  \beta(y_i)$ iff  $ \alpha^b(y_i) =  \beta^b(y_i)$ for all $i=1, \ldots, k$.
Also, $ \alpha(y_i) >  \beta(y_i)$ iff  $ \alpha^b(y_i) >  \beta^b(y_i)$ for all $i=1, \ldots, k$. Thus, $\alpha \succeq _\pi \beta$ iff $\alpha^b \succeq _\pi \beta^b$. So $\pi\models\phi$ iff $\pi\models\phi^b$.
\end{proof}

\end{document}